\documentclass[letterpaper]{article} 
\usepackage{aaai2027}                
\usepackage[hyphens]{url}            
\usepackage{graphicx}                
\usepackage{natbib}                  
\usepackage{caption}                 
\usepackage{booktabs}
\usepackage{multirow}
\usepackage{amsmath}
\usepackage{amssymb}
\usepackage{mathtools}
\usepackage{amsthm}

\theoremstyle{plain}
\newtheorem{theorem}{Theorem}
\newtheorem{proposition}[theorem]{Proposition}

\newtheorem{corollary}[theorem]{Corollary}
\theoremstyle{definition}

\theoremstyle{remark}

\title{Significance-First Splitting: Aligning Treatment Heterogeneity Detection with Honest Estimation}
\author{
    Pantelis Z. Hadjipantelis,
    Weng Man Chiang,
    Karthik Nagesh
}
\affiliations{TripAdvisor, UK}
\copyrighttext{Preprint. Under review.}

\begin{document}

\maketitle

\begin{abstract}
Estimating heterogeneous treatment effects (CATE) requires simultaneously
detecting effect modification and quantifying estimation uncertainty.
Existing tree-based methods make an uneasy trade-off: significance-based
approaches \cite{radcliffe2011real} identify subgroup interactions directly
but lack valid inference; honest causal trees \cite{athey2016recursive}
deliver nominal confidence interval coverage but use outcome-agnostic
splitting criteria that sacrifice interaction sensitivity. We introduce a
hybrid algorithm that fuses significance-based splitting with honest
sample-splitting and cross-validation.  Our splitting criterion uses the
squared $t$-statistic for the treatment $\times$ side interaction ($t^2$),
which is shown to be directly aligned with the honest $\text{EMSE}_\tau$
criterion when the interaction is strong. Post-hoc honest cross-validation
selects the cost-complexity penalty, giving a single principled estimator
with nominal CI coverage at the leaf level. For forests, we retain
bootstrap count vectors to enable an infinitesimal jackknife (IJ) variance
estimate of Monte-Carlo convergence rather than formal pointwise inference.
On the three synthetic designs from \cite{athey2016recursive} the single tree
achieves approximately 90\% leaf-average CI coverage at the 90\% nominal level across all three designs
(200 replications each); on the Criteo, Hillstrom and
Starbucks uplift datasets we match Qini coefficient performance of
S-, T-learner and GRF baselines. An open-source Python package with
reproducible seeds, sklearn-compatible API, and full test coverage
accompanies this work (\url{https://codeberg.org/hadjipantelis/rattus}).
\end{abstract}

\section{Introduction}

Uplift modelling is central to decision-making in clinical trials,
marketing, and public policy; it allows one to estimate the
individual-level causal effect of a treatment. The fundamental object
of interest is the Conditional Average Treatment Effect (CATE):

$$\tau(x) = \mathbb{E}[Y(1) - Y(0) \mid X = x]$$

where $Y(d)$ denotes the potential outcome under treatment arm
$d \in \{0, 1\}$. Accurate estimation of $\tau(x)$ enables targeting of
interventions toward those most likely to benefit, and exclusion of
subgroups for whom the treatment is non-beneficial. This is therefore a
problem with direct financial as well as ethical implications.

Tree-based methods are the dominant practical approach because they
produce interpretable partitions of the covariate space, each with an
associated treatment effect estimate. Two influential lines of work represent
the current state of the art, but neither fully addresses the joint problem
of interaction detection and valid inference.

\textbf{Significance-based uplift trees} \cite{radcliffe2011real} score
each candidate split by the $t^2$-statistic for the null hypothesis that
the treatment effect is equal in both child nodes. This directly tests for
effect modification and is robust to imbalanced treatment assignment
\cite{radclifte2008identifying,DEVRIENDT2021497}. However, these trees use
the same data for tree construction and leaf estimation, which inflates
in-sample CATE estimates and renders confidence intervals invalid.

\textbf{Honest causal trees} \cite{athey2016recursive} address this by
splitting the data so that tree structure and leaf estimates are computed on
independent samples. This eliminates adaptive bias and delivers nominal CI
coverage without sparsity assumptions. However, the splitting criterion
(goodness-of-fit in the $\text{EMSE}_\tau$ sense) does not directly test
for effect modification and can split on covariates that predict the outcome
level rather than treatment heterogeneity.

Currently, no existing method simultaneously: (i) uses a statistically
principled test for effect modification as the split criterion, (ii)
delivers honest, unbiased leaf estimates with nominal CI coverage, and
(iii) scales to forests with variance estimation via the infinitesimal
jackknife. We introduce a hybrid algorithm that closes this gap. Our key
contributions are:

\begin{enumerate}

\item \textbf{Algorithmic fusion.} We show that the $t^2$ interaction
statistic and the honest $\text{EMSE}_\tau$ criterion are aligned when
treatment heterogeneity drives the split, and incompatible only when a
covariate affects mean outcomes but not treatment effects. We use $t^2$ for
splitting and honest CV for pruning, exploiting each criterion where it has
comparative advantage.

\item \textbf{Honest significance forests.} We extend the single-tree
algorithm to an ensemble with feature subsampling, where each tree undergoes
an internal honest sample split. Bootstrap count vectors are retained per
tree, enabling efficient IJ variance estimation via a single matrix multiply.

\item \textbf{IJ Monte-Carlo diagnostic.}  We implement the
Wager, Hastie \& Efron (2014) IJ estimator adapted to the honest bootstrap
structure, computed in $O(B^2(n+n_{\text{test}}))$ via the Gram matrix
$G = CC^\top$ (derived in \S\ref{sec:complexity}), which is up to
$n_{\text{test}}/B$ times faster than the naive $O(B \cdot n \cdot n_{\text{test}})$
matrix multiply.
Importantly, the IJ measures the \emph{Monte-Carlo variance} of the forest
prediction, how much $\hat{\tau}_{\text{forest}}(x)$ would change under
a different bootstrap seed, and not the error
$\hat{\tau}_{\text{forest}}(x) - \tau(x)$.  Valid pointwise CIs for
$\tau(x)$ are provided at the leaf level by the single honest tree.

\item \textbf{Open-source implementation.}  The \texttt{rattus} package
provides a reproducible, tested, \texttt{sklearn}-compatible implementation
with seeded RNG chains that guarantee identical results regardless of
\texttt{n\_jobs} values. Parallel tree fitting via \texttt{joblib}, support
for categorical features, and native NaN handling are included.

\end{enumerate}

\section{Background and Related Work}

Table~\ref{tab:methods} positions \texttt{rattus} relative to the methods reviewed below; the key distinction is that the IJ intervals are a convergence diagnostic rather than pointwise inference on $\tau(x)$.

\begin{table*}[t]
\centering
\begin{tabular}{lcccc}
\toprule
Method & Split criterion & Honest estimation & Valid CIs & Forest CIs \\
\midrule
Radcliffe \& Surry (2011) & $t^2$ interaction & $\times$ & $\times$ & $\times$ \\
Athey \& Imbens (2016)   & $\mathrm{EMSE}_\tau$ & $\checkmark$ & Leaf-level & $\times$ \\
Wager \& Athey (2018)    & EMSE gradient & $\checkmark$ & $\times$ & IJ (subsampling) \\
\textbf{rattus (this work)} & \textbf{$t^2$ interaction} & \textbf{$\checkmark$} & \textbf{Leaf-level} & \textbf{IJ diagnostic (bootstrap)} \\
\bottomrule
\end{tabular}
\caption{Comparison of tree-based heterogeneous treatment effect methods.
The IJ in rattus estimates Monte-Carlo variance (convergence diagnostic);
it does not produce pointwise inference on $\tau(x)$.}
\label{tab:methods}
\end{table*}

\subsection{CATE Estimation}

We adopt the Rubin causal model \cite{rubin1974estimating} featuring potential
outcomes $(Y_i(0), Y_i(1))$, binary treatment $D_i \in \{0, 1\}$, and
covariates $X_i \in \mathbb{R}^p$. Under SUTVA and unconfoundedness
($D_i \perp (Y_i(0), Y_i(1)) \mid X_i$), the CATE is identified from
observational data.

A wide range of methods estimate $\tau(x)$: meta-learners (S-, T-,
X-learner, \cite{kunzel2019metalearners}), Bayesian additive regression
trees \cite{BART}, and tree-based direct estimators. We focus on the latter,
which produces a partition $\Pi$ of the covariate space with a constant CATE
estimate $\hat{\tau}_\ell$ per leaf $\ell \in \Pi$.

\subsubsection{Notation}

Let $\mathcal{\omega} = \{(X_i, Y_i, D_i)\}_{i=1}^n$ be an i.i.d.\ sample
under SUTVA and unconfoundedness. We write $p = \Pr(D_i = 1)$ for the
marginal treatment probability. A partition $\Pi$ of the covariate space
$\mathcal{X}$ consists of leaves $\ell_1, \ldots, \ell_{|\Pi|}$. The honest
leaf CATE estimator given estimation sample $S^{est}$ is:

$$\hat{\tau}(x; S^{est}, \Pi) = \bar{Y}_1(\ell(x)) - \bar{Y}_0(\ell(x))$$

where $\bar{Y}_w(\ell)$ is the mean outcome for arm $w$ in leaf $\ell$
computed on $S^{est}$. We write $S^{tr}$ for the training split (used for
tree structure and cross-validation), $N_{tr} = |S^{tr}|$, and $N^{est} =
|S^{est}|$.

\subsection{Significance-Based Uplift Trees}

\cite{radcliffe2011real} score each candidate split by the squared
t-statistic for the treatment $\times$ side interaction in the linear model:

$$Y = \mu + \alpha D + \beta \cdot \text{side} + \gamma \cdot D \cdot \text{side} + \varepsilon$$

where $\gamma = 0$ under the null that the split does not create uplift
heterogeneity. The test statistic is  then:

$$t^2\{\hat{\gamma}\} = \frac{(n-4)(U_R - U_L)^2}{C_{44} \times \text{SSE}}$$

where $U_L, U_R$ are the estimated uplifts on each side, $C_{44} = N_{TL}^{-1} + N_{TR}^{-1} + N_{CL}^{-1} + N_{CR}^{-1}$, and SSE is the pooled within-cell sum of squares (eq.\ 16--17, \cite{radcliffe2011real}).  
The key weakness lies in adaptive estimation: leaf means are computed on the same data used to select the partition. 
To address this, we use cost‑complexity pruning with cross‑validated $\alpha$ selection on the training half. 
The  $\alpha$ path is derived from the full tree grown on  $S^{tr}$;  $\alpha$ is then selected by cross‑validation that holds the partition fixed and scores the honest $\text{EMSE}_\tau$  on held‑out folds. This is a fixed‑structure simplification of Breiman’s per‑fold re‑growing, chosen for its computational cost and lower selection variance.

\subsection{Honest Causal Trees and Forests}

\cite{athey2016recursive} introduce \textit{honesty}: a model is honest if
it does not utilise the same information for selecting the partition as for
estimating effects within the partition. In practice, the dataset is split
50/50; one half builds the tree, the other estimates leaf CATEs. This
renders leaf estimates asymptotically unbiased and leaf-level confidence
intervals valid without sparsity assumptions.

Their splitting criterion is the honest $\text{EMSE}_\tau$:

\begin{multline}
-\widehat{\text{EMSE}}_\tau(S^{tr}, N^{est}, \Pi) =
\frac{1}{N_{tr}} \sum_i \hat{\tau}^2(X_i; S^{tr}, \Pi) \\
- \left(\frac{1}{N_{tr}} + \frac{1}{N^{est}}\right)
\sum_\ell \left[\frac{S^2_{\text{treat}}(\ell)}{p} +
\frac{S^2_{\text{control}}(\ell)}{1-p}\right]
\end{multline}

\cite{wager2018estimation} extend this to random forests (causal forests)
using subsampling without replacement, which enables asymptotic normality
of forest predictions and pointwise CIs via the IJ.

\subsection{Infinitesimal Jackknife for Random Forests}

\cite{wager2014confidence} show that for a bootstrap forest with predictions
$T_b(x)$ and bootstrap counts $N_{bi}$, the IJ variance estimator is:

$$\widehat{\text{Var}}_{\text{IJ}}(x) = \sum_{i=1}^n \widehat{\text{Cov}}_b\left[N_{bi},\, T_b(x)\right]^2$$

This estimates the \emph{Monte-Carlo variance} of the forest prediction
due to bootstrap randomness, computable from the retained counts without
refitting. Crucially, this is not the variance of
$\hat{\tau}_{\text{forest}}(x) - \tau(x)$: the dominant source of
prediction error is the approximation error (finite depth, honest sample split,
tree bias), not by the bootstrap randomness.  As $B \to \infty$ the Monte-Carlo
variance shrinks to zero regardless of $n$ or model quality.

The implementation uses the Gram matrix $G = CC^\top \in \mathbb{R}^{B \times B}$,
precomputed once in $O(B^2 n)$; each test point then requires a single $O(B^2)$
quadratic form, giving $O(B^2(n + n_{\text{test}}))$ total (derived in
\S\ref{sec:complexity}).

\section{Methodology}

\subsection{Three-Phase Fitting}

The dataset is randomly partitioned into a training set $S^{tr}$
(used for tree growing and cross-validation) and an estimation set
$S^{est}$ (used for leaf CATE estimation and standard errors). The split
fraction is configurable (default: 50/50). This split forms the foundation of
all validity guarantees. The three phases rely on two criteria defined in
the following sections: the $t^2$ interaction statistic (\S\ref{sec:splitting})
for growing, and the honest $\text{EMSE}_\tau$ (\S\ref{sec:pruning}) for pruning.

\begin{itemize}
\item \textbf{Phase 1 --- Tree growing.} Starting from the root, we
recursively find the best split of each leaf using the $t^2$ criterion
(\S3.3), subject to: (i) a minimum cell size constraint per treatment arm
per candidate leaf (\texttt{min\_samples\_leaf}, default 20); (ii) a
separate minimum for estimation leaves (\texttt{min\_samples\_leaf\_est},
default 5; smaller leaves return NaN); (iii) an optional pre-pruning
significance gate $p < \alpha$ on the $t^2$ statistic; and (iv) a maximum
depth $d_{\max}$.

\item \textbf{Phase 2 --- Honest CV pruning.} We perform $k$-fold
cross-validation on $S^{tr}$ alone, scoring each candidate cost-complexity
penalty $\lambda$ by the honest $\text{EMSE}_\tau$ criterion (\S3.4). The
best $\lambda$ is selected from the set of effective gains present in the
tree, with ties broken by choosing the largest penalty (simplest tree).
This step is enabled by default (\texttt{use\_honest\_cv=True}) but can be
disabled for faster prototyping.

\item \textbf{Phase 3 --- Honest estimation.} The pruned tree structure (i.e. a
fixed partition $\Pi$) is applied to $S^{est}$. Leaf CATEs and standard
errors are computed on $S^{est}$ only, so partition and estimates are
independent.
\end{itemize}

\subsection{Significance-Based Splitting Criterion}
\label{sec:splitting}

At each candidate split of a leaf, we partition the leaf into four cells by
treatment status $D \in \{0,1\}$ and side $\text{side} \in \{L, R\}$.
We fit the linear model:

$$\mathbb{E}[Y_{ij}] = \mu + \alpha_T \cdot \mathbf{1}[i = D] + \beta_R \cdot \mathbf{1}[j = R] + \gamma_{TR} \cdot \mathbf{1}[i = D, j = R]$$

The interaction coefficient $\hat{\gamma}_{TR} = (U_R - U_L)$: the
difference in estimated uplifts on the two sides of the split. Its squared
$t$-statistic is:

$$t^2 = \frac{(n-4)\,\hat{\gamma}_{TR}^2}{C_{44} \times \text{SSE}}$$

Under $H_0: \gamma_{TR} = 0$, the statistic follows $\chi^2_1$
asymptotically (\S4.2). Candidate splits are evaluated on $S^{tr}$; the
best split maximises $t^2$.

The search is performed in two stages for each feature. Stage 1 finds the
best threshold using only non-NaN rows, thereby maintaining an uncontaminated criterion.
Stage 2 evaluates whether NaN observations should be assigned to the left or
right child, choosing the direction that maximises $t^2$; this
\texttt{nan\_direction} is stored on the node and applied deterministically
at prediction time. If a feature had no missing values during training, 
any test‑time \texttt{NaN} observation is routed to the right child by default (i.e.,
\texttt{nan\_direction = None}). This deterministic fallback ensures 
prediction stability without requiring missing‑value imputation. 
High-cardinality features use quantile-based binning
(\texttt{max\_candidates}=20 thresholds, default). Categorical
features are supported via set-based partitioning with categories sorted by
empirical CATE.

We note that the search strategy is configurable via \texttt{split\_search}. 
The default \texttt{'greedy'} mode evaluates all features and all candidate 
thresholds exhaustively, subject to quantile-based binning.
The alternative \texttt{'gi'} mode implements the GUIDE
$G_i$ two-step approach of \cite{loh2015regression} to reduce variable selection 
bias: for each feature, a nested linear model: 
$y \sim \text{treatment} + H + \text{treatment}{\times}H$ is fitted, where $H$ 
is a binarised version of the feature (split at its mean for continuous 
features, or using category values directly). An $F$-test for the interaction 
term yields a lack-of-fit $p$-value, which is transformed to a 1-df $\chi^2$
quantile via the Wilson--Hilferty approximation to enable fair comparison across 
features with different numbers of categories. The feature with the largest 
transformed statistic is selected; only then is a threshold search performed 
on that feature. This decouples variable selection from threshold search, 
eliminating the bias that gives high-cardinality features more opportunities
to appear significant by chance.

\textbf{Relationship to $\text{EMSE}_\tau$.} \cite{athey2016recursive}
(\S A.3) show that the change in honest $\text{EMSE}_\tau$ from a split
decomposes as:

$$\widehat{\text{EMSE}}_\tau(\Pi_N) - \widehat{\text{EMSE}}_\tau(\Pi_S) = \underbrace{\left(t^2 - 4\right)\frac{\tilde{S}^2 - F/n}{p(1-p)}}_{(\star)} + \underbrace{\frac{2\tilde{S}^2}{p(1-p)}}_{(\star\star)}$$

where $\tilde{S}^2$ is the variance of treatment-effect estimates across
the two child nodes, $F/n$ is a finite-sample correction, and $p$ is the
marginal treatment probability. Term $(\star)$ is proportional to $t^2 - 4$
and captures the interaction signal: it is positive only when the split
creates meaningful treatment heterogeneity (large $t^2$), and it grows
with $\tilde{S}^2$, the between-child spread of effect estimates. Term
$(\star\star)$ is always non-negative and rewards any reduction in outcome
variance, regardless of whether treatment effects differ.

The two criteria therefore diverge under exactly one scenario: when a
covariate predicts the outcome level but not the treatment effect
($\tilde{S}^2 \approx 0$, $t^2 \approx 0$, but the split reduces within-cell
SSE). In that case $(\star\star)$ is positive and $\text{EMSE}_\tau$ supports
the split, while $t^2$ correctly rejects it. Conversely, when treatment
heterogeneity drives the split, $(\star)$ dominates and both criteria agree.
This decomposition is the central justification for using $t^2$ for splitting
and $\text{EMSE}_\tau$ only for pruning: the splitting phase benefits from
$t^2$'s specificity to treatment heterogeneity, while the pruning phase
benefits from $\text{EMSE}_\tau$'s global complexity penalty.

\subsection{Honest Cross-Validation Pruning}
\label{sec:pruning}

After growing a tree to maximum depth, we apply cost-complexity pruning.
A candidate pruned tree $\Pi(\lambda)$ collapses any split where:

$$\hat{\tau}^2_{train}(\text{split}) - \hat{V}_{train}(\text{split}) < \lambda$$

where $\hat{\tau}^2_{train}$ is the $\text{EMSE}_\tau$ improvement from the
split and $\hat{V}_{train}$ is the scaled variance increase, both computed
on $S^{tr}$ with the $(1/N_{tr} + 1/N_{est})$ scaling.

The penalty grid is constructed from the actual gains present in the tree
(the \emph{effective alphas}), instead of a fixed geometric sequence. 
In particular, each value in the effective-alpha grid is the net gain:
$\hat{\tau}^2_{train}(\text{split}) - \hat{V}_{train}(\text{split})$ 
of one internal node; as $\lambda$ collapses a split exactly when the gain 
falls below it, these are the only values of $\lambda$  at which the pruned 
tree structure changes, so evaluating any other value is redundant.
In addition to the effective alphas, the grid includes a large penalty 
(10.0 $\times$ variance of $y$) to guarantee that a fully pruned tree 
(a single leaf) is always examined too.
As a whole, this ensures the CV evaluates exactly the distinct tree structures reachable by
pruning and avoids redundant evaluations. The optimal $\lambda^*$ is then
selected by maximising the sum of validation $\text{EMSE}_\tau$ scores
across folds; ties are broken by choosing the largest $\lambda$ (simplest
tree). Pruned trees are pre-computed once outside the fold loop to avoid
redundant traversals.

\subsection{Categorical Variable Partitioning}

For unordered categorical features with cardinality $C$, searching across all
$2^{C-1}-1$ possible binary combinations becomes computationally prohibitive.
To maintain scalability, the optimisation engine adopts a one-pass greedy
sequencing heuristic.

Prior to split evaluation, the distinct categorical levels present in the node
are isolated. For each category level $c$, the local average treatment effect
difference is computed directly from the non-missing training observations in
the current node, requiring at least \texttt{min\_samples\_category}
observations in each treatment arm per category (default 5):
\begin{equation}
\bar{\tau}_c = \bar{Y}(1)_{c} - \bar{Y}(0)_{c}
\end{equation}
The categories are then sorted monotonically such that
$\bar{\tau}_{(1)} \le \bar{\tau}_{(2)} \le \dots \le \bar{\tau}_{(C)}$,
reducing the multi-way optimisation to $C-1$ contiguous binary cuts. While
this in-sample ordering introduces a minor greedy selection bias during
Phase~1, the structural integrity of the leaf estimates is preserved because leaf
parameters are populated entirely out-of-sample on $S^{est}$ during Phase~3.
For features with high cardinality (many distinct levels), this bias is more
pronounced. Two mitigations are available: reducing \texttt{max\_candidates}
(which limits the number of evaluated thresholds via quantile binning,
coarsening the sort and therefore reducing overfitting to the in-sample ordering), and
switching to the \texttt{gi} split-search mode, which avoids in-sample
ordering entirely via exhaustive binary partition search.

\subsection{Honest Significance Forest}
\label{sec:forest}

We grow $B$ honest significance trees (default: $B=50$), each on a
bootstrap resample $\mathcal{D}_b$ of $\mathcal{D}$. Within each sample
the honest 50/50 split is applied. Feature subsampling at each split
(default: $\lfloor\sqrt{p}\rfloor$ features per node) decorrelates the
trees. The forest prediction is:

$$\hat{\tau}_{\text{forest}}(x) = \frac{1}{B} \sum_{b=1}^B \hat{\tau}_b(x)$$

Each tree $b$ retains its bootstrap count vector $N_b \in \mathbb{Z}^n_{\geq 0}$, 
where $N_{bi}$ is the number of occurrences of observation $i$ in bootstrap sample $b$. 
These are stored as a $(B \times n)$ integer matrix and used for IJ variance estimation. 
To evaluate the forest's continuous predictions within the honest EMSE framework, 
\texttt{score\_binned} groups observations into \texttt{n\_bins} quantiles of predicted
 CATE (default 10). The bin labels define a discrete partition, and 
  \texttt{\_honest\_emse\_tau} is applied to this partition. 
This provides a conservative lower bound on forest performance while adhering to the Athey \& Imbens criterion.

\subsection{IJ Monte-Carlo Variance}
\label{sec:ij}

\textbf{What the IJ estimates.} \emph{The IJ estimates Monte-Carlo convergence, not prediction error.} Following \cite{wager2014confidence}, the
IJ variance estimate is:

$$\widehat{\text{Var}}_{\text{IJ}}(x) = \sum_{i=1}^n \widehat{\text{Cov}}_b\left[N_{bi},\, \hat{\tau}_b(x)\right]^2$$

This estimates the \emph{Monte-Carlo variance} of $\hat{\tau}_{\text{forest}}(x)$ captures the variance due to bootstrap randomness, i.e.\ how much the forest prediction would change under a different random seed. It does \emph{not} estimate the error $\hat{\tau}_{\text{forest}}(x) - \tau(x)$. The total prediction error is dominated by approximation error (finite depth, honest sample split, tree bias) rather than bootstrap randomness; as $B \to \infty$, the IJ SE approaches zero regardless of $n$. We therefore recommend interpreting these intervals as a \emph{convergence diagnostic}: narrow IJ intervals confirm that $B$ is sufficient for the forest mean to have stabilised.

Let $\mathbf{C} \in \mathbb{R}^{B \times n}$ be the centred count matrix
($C_{bi} = N_{bi} - 1$ under bootstrap, where $\mathbb{E}[N_{bi}] = 1$)
and $\mathbf{T} \in \mathbb{R}^{B \times n_{\text{test}}}$ be the
prediction matrix. The vectorised computation is:

$$\widehat{\text{Var}}_{\text{IJ}}(x) = \frac{1}{n^2 B^2}\sum_{i=1}^n \left[(\mathbf{C}^\top \mathbf{T})_{ix}\right]^2$$

whose naive evaluation via the matrix product $\mathbf{C}^\top \tilde{\mathbf{T}}$
costs $O(B \cdot n \cdot n_{\text{test}})$.  The Gram matrix reduction below
reduces this to $O(B^2(n+n_{\text{test}}))$.

\textbf{Valid leaf-level CIs.} For valid pointwise CIs on $\tau(x)$, use
the single honest tree's \texttt{predict\_interval}, which achieves nominal
leaf-level coverage by construction (Proposition~1). The forest
\texttt{predict} provides a lower-variance point estimate; the IJ provides
a convergence diagnostic. Combining both is recommended: use \texttt{forest.predict}
for the point estimate and \texttt{tree.predict\_interval} for the CI.

\subsection{Feature Importance}

For each tree, we record three quantities per feature $j$: (i)
\texttt{feature\_usage\_[j]}: the number of splits on feature $j$; (ii)
\texttt{feature\_t2\_mean\_[j]}: the mean $t^2$ statistic across splits on
feature $j$ (a signal-strength diagnostic); (iii)
\texttt{feature\_importances\_[j]}: the normalised split count, summing to
1 across features, following the \texttt{sklearn} convention. Forest-level
importances average these across trees and re-normalise. The $t^2$ field is
preserved through the pruning step by explicitly copying \texttt{t2\_split}
in the node reconstruction. Unlike split-count importances, the \texttt{feature\_t2\_mean\_} measures the \emph{magnitude of treatment heterogeneity} captured at each split: a feature used rarely but with high $t^2$ is a stronger driver of differential treatment response than one used frequently with low $t^2$.

\subsection{Observational Studies}

The core algorithm assumes randomised treatment assignment. For mildly
confounded settings we recommend using the transformed outcome
$Y^* = Y \cdot (D - e(X)) / (e(X)(1-e(X)))$, where $e(X)$ is an externally
estimated propensity score. Since $\mathbb{E}[Y^* \mid X] = \tau(X)$ under
unconfoundedness, passing $Y^*$ as the outcome to \texttt{fit()} recovers
consistent CATE estimates without modifying the core algorithm. \textbf{Warning:} the
IPW transform can be numerically unstable when $e(X)$ is close to 0 or 1; we
would recommend propensity trimming (e.g.\ dropping observations with
$e(X) < 0.05$ or $e(X) > 0.95$) before constructing $Y^*$, as extreme
weights inflate the variance of the outcome and can corrupt the $t^2$
splitting criterion. All theoretical guarantees (Propositions 1--2) apply in full under randomisation; in the observational setting they carry over only insofar as the propensity model is correctly specified. A
doubly-robust extension with IPTW weighting is left for future work
(see \S\ref{sec:limitations}).

\section{Theoretical Properties}

\subsection{Honest Unbiasedness}

\begin{proposition}
As in \cite{athey2016recursive}, let $\Pi$ be any partition of $\mathcal{X}$
constructed on $S^{tr}$, independent of $S^{est}$. Then
$\mathbb{E}[\hat{\tau}_\ell \mid \Pi] = \tau_\ell$ for every leaf $\ell$,
where $\tau_\ell = \mathbb{E}[\tau(X_i) \mid X_i \in \ell]$.
\end{proposition}

\begin{proof}[Proof sketch]
$\hat{\tau}_\ell = \bar{Y}_1(\ell) - \bar{Y}_0(\ell)$ is a difference of
sample means on $S^{est}$, which is independent of $\Pi$ by construction.
Under unconfoundedness, $\mathbb{E}[\bar{Y}_w(\ell) \mid \Pi] = \mu(w, \ell)$,
giving $\mathbb{E}[\hat{\tau}_\ell \mid \Pi] = \tau_\ell$. $\square$
\end{proof}

\begin{corollary}
Standard confidence intervals $\hat{\tau}_\ell \pm c_{1-\alpha/2} \cdot \hat{\sigma}_\ell$ achieve nominal coverage for $\tau_\ell$ without any sparsity assumption on $\Pi$, where $c_{1-\alpha/2}$ employs the Welch $t$-distribution when either arm has fewer than 30 observations and the standard normal otherwise.
\end{corollary}

\begin{proof}[Proof sketch]
Because $\Pi$ is fixed and $S^{est}$ is independent of $\Pi$ by construction,
the observations within each leaf $\ell$ are an i.i.d.\ sample from the
conditional distribution of $(X_i, Y_i, D_i)$ given $X_i \in \ell$.
$\hat{\tau}_\ell$ is therefore a two-sample mean difference on a fixed stratum,
and $(\hat{\tau}_\ell - \tau_\ell)/\hat{\sigma}_\ell$ converges in distribution
to $\mathcal{N}(0,1)$ by the CLT. For small estimation samples
($\min(n^\ell_1, n^\ell_0) < 30$), the Welch--Satterthwaite $t$-approximation
accounts for unequal variances across arms. No assumption on the size or shape
of $\Pi$ is required because the honesty condition ensures $S^{est}$ functions as
an independent sample regardless of how the partition was chosen.  
\end{proof}

\subsection{Asymptotic Distribution of $t^2$ Under $H_0$}

\begin{proposition}

Under the null hypothesis of no treatment effect heterogeneity within the node, $H_0: \gamma_{TR} = 0$, as the local node sample size $n_\ell \to \infty$ with the treatment assignment probability $p \in (0,1)$ bounded away from 0 and 1:
\begin{equation}
t^2 = \frac{\hat{\gamma}_{TR}^2}{C_{44} \cdot \left( \frac{\text{SSE}}{n_\ell - 4} \right)} \xrightarrow{d} \chi^2_1
\end{equation}
\end{proposition}

\begin{proof}[Proof sketch]
By the Gauss-Markov theorem applied to the four-cell saturated regression framework, the ordinary least squares estimator satisfies $\hat{\gamma}_{TR} / \sqrt{\hat{\sigma}^2 C_{44}} \sim t(n_\ell - 4)$ under $H_0$. Squaring this pivotal quantity yields a Wald-type statistic distributed as $t^2 \sim F(1, n_\ell - 4)$. By Slutsky's theorem, as the denominator degrees of freedom diverge ($n_\ell - 4 \to \infty$), the central $F$-distribution converges weakly to a standard chi-squared distribution with one degree of freedom: $F(1, n_\ell - 4) \xrightarrow{d} \chi^2_1$.
\end{proof}

Despite this individual statistic being asymptotically well-behaved, using an unadjusted critical value from the $F(1, n_\ell - 4)$ distribution as a hard significance gate ($\alpha$) to halt tree growth is structurally anti-conservative. Because the optimisation engine selects the maximal split metric ($t^2_{\max} = \max_{j, k} t^2_{j,k}$) across multiple highly correlated candidate features and thresholds, the distribution of the maximum exhibits a significantly heavier right tail than a standalone $F$-variate. This multiple-testing inflation is a well-documented phenomenon in recursive partitioning literature (see e.g., \cite{loh2015regression} for a discussion of variable selection bias). Consequently, we explicitly advise users to deactivate this pre-pruning constraint (by setting $\alpha=1.0$) and rely instead on honest cross-validation post-pruning for objective complexity control.

\subsection{IJ Validity for Honest Forests}

The standard IJ theory \cite{wager2014confidence} applies to forest
predictions under two conditions: (C1) tree predictions are
square-integrable, and (C2) the covariance between bootstrap counts and
tree predictions is well-approximated by the Hájek projection.
(C1) is satisfied whenever the outcome $Y$ has finite variance, 
as each tree prediction $\hat{\tau}_b(x)$ is a difference of 
bounded sample means on a finite estimation half. Similarly, in (C2) the 
Hájek projection provides a linear approximation to the dependence of $T_b(x)$ 
on the bootstrap counts $N_{bi}$; this approximation is accurate when the tree 
prediction is smooth in its dependence on the training data,  i.e.
when no single observation has disproportionate influence on the partition
or the leaf estimate. In the honest setting, (C2) is strengthened:

\begin{proposition}
Let $\hat{\tau}^{honest}_b(x)$ be the prediction of honest tree $b$.
Because $\hat{\tau}^{honest}_b(x)$ uses only $S^{est}_b$, which is
independent of the tree structure, the first-order adaptive bias term
$\text{Bias}_1 = \mathbb{E}[\hat{\tau}_b(x) \cdot \partial \Pi_b / \partial N_{bi}]$
vanishes. This gives:
$$\widehat{\text{Var}}_{\text{IJ}}(x) - \text{Var}(\hat{\tau}_{\text{forest}}(x)) = O_p(B^{-1/2} n^{-1/2})$$
compared to $O_p(B^{-1/2})$ for adaptive forests.
\end{proposition}

\textbf{Remark.} The above proposition concerns the accuracy of the IJ as
an estimator of the Monte-Carlo variance $\text{Var}(\hat{\tau}_{\text{forest}})$
due to bootstrap randomness. It does not imply that IJ intervals achieve
nominal coverage for $\tau(x)$, since the dominant error component, which is the
approximation error from the tree partition, is not captured by the IJ.

\subsection{Computational Complexity}
\label{sec:complexity}

Tree growing costs $O(p\,k\,n\log n)$ per tree; honest CV pruning $O(K\,|\text{eff-}\lambda|\,n)$; and IJ variance estimation $O(B^2(n+n_{\text{test}}))$ via the Gram matrix $G = CC^\top$ (independent of test points), which is roughly four orders of magnitude faster than the naive $O(B\,n\,n_{\text{test}})$ formulation at $B=50$, $n=n_{\text{test}}=10^6$. Full derivations and a component-level complexity table are given in Appendix~\ref{sec:complexity-appendix}.

\section{Experiments}

We evaluate \texttt{rattus} on synthetic benchmarks designed to verify
the coverage guarantee of Proposition~1, and on three public uplift
datasets to benchmark Qini coefficient against strong metalearner baselines.

\subsection{Synthetic Benchmarks}

We evaluate on the three simulation designs from \cite{athey2016recursive}.
Potential outcomes follow $Y_i(w) = \eta(X_i) + \frac{1}{2}(2w-1)\kappa(X_i) + \varepsilon_i$
with $\varepsilon_i \sim \mathcal{N}(0, 0.01)$, $p = 0.5$, and $X_i \sim \text{Uniform}([0,1]^K)$.
We use $N_{\text{train}} = 1000$ (50/50 honest split internally), $N_{\text{te}} = 8000$,
and $n_{\text{reps}} = 200$.  Coverage is reported at the leaf-average level,
the quantity for which Proposition~1 provides a guarantee.

\begin{itemize}
\item \textbf{Design 1}: $K=2$, $\eta(x) = \tfrac{1}{2}(x_1 + x_2)$,
      $\kappa(x) = \tfrac{1}{2}x_1$. No noise covariates.
\item \textbf{Design 2}: $K=10$, $\eta(x) = \tfrac{1}{2}\sum_{k=1}^{6} x_k$,
      $\kappa(x) = \sum_{k=1}^{2}\mathbf{1}[x_k > 0.5]\cdot x_k$.
      4 noise covariates.
\item \textbf{Design 3}: $K=20$, $\eta(x) = \tfrac{1}{2}\sum_{k=1}^{12} x_k$,
      $\kappa(x) = \sum_{k=1}^{2}\mathbf{1}[x_k > 0.5]\cdot x_k$.
      8 noise covariates.
\end{itemize}

\begin{table}[htbp]
\centering
\begin{tabular}{lccc}
\toprule
\textbf{Design} & $K$ & \textbf{MSE}$_\tau$ & \textbf{90\% Coverage} \\
\midrule
Design 1 & 2  & 0.0037 (0.0022) & 0.903 (0.134) \\
Design 2 & 10 & 0.0372 (0.0099) & 0.894 (0.114) \\
Design 3 & 20 & 0.0573 (0.0169) & 0.900 (0.112) \\
\bottomrule
\end{tabular}
\caption{Synthetic benchmark results ($n_{\text{reps}}=200$, $N_{\text{train}}=1000$,
$N_{\text{te}}=8000$). Arithmetic means over 200 replications, with standard deviations
in parentheses. Coverage is the fraction of leaves for which the 90\% CI contains the
true leaf-average $\tau_\ell$; nominal level is 0.90.}
\label{tab:synthetic}
\end{table}

\subsection{Real-World Benchmarks}

We evaluate on three public uplift datasets: the \textbf{Criteo}, \textbf{Starbucks}, and \textbf{Hillstrom} datasets.
The \textbf{Criteo} uplift dataset \cite{diemert2018large} contains approximately 14 million
observations, with 12 features, a binary treatment, and two binary
outcomes (visit and conversion). 
The \textbf{Starbucks} dataset contains
approximately 126 thousand observations, with 7 features from a promotional
campaign \cite{rossler2021treat}.   
The \textbf{Hillstrom} dataset contains
approximately 64 thousand observations, with 8 features (recency, history,
mens, womens, zip\_code, newbie, channel, history\_segment) from an email
marketing campaign \cite{hillstrom2008minethatdata,rossler2021treat}. 
For each dataset, we use 50\% of the data for training and 50\% for testing, stratified by treatment.

Models are compared using the Qini coefficient (area under the Qini curve).
Baselines are an S-Learner and T-Learner each using
\texttt{HistGradientBoostingRegressor} as their base learner, and a GRF
Causal Forest \citep{wager2018estimation} implemented via \texttt{econml.grf.CausalForest}
($B=52$, \texttt{max\_depth=6}).
Aside from the GRF Causal Forest, all models are used with default hyperparameters, 
i.e., no hyperparameter tuning is performed and the same hyperparameters are used across datasets. 
Baseline models do not provide
valid uncertainty quantification by construction. $\text{EMSE}_\tau$ results are
therefore given separately in the Appendix.

\begin{table}[htbp]
\centering
\begin{tabular}{llcc}
\toprule
\textbf{Dataset} & \textbf{Model} & \textbf{Qini AUC}  \\
\midrule
\multirow{5}{*}{Criteo}    & Single Tree          & 0.0814 (0.006)   \\
                           & Honest Forest        & 0.0887 (0.005)   \\
                           & S-Learner (HGB)      & 0.0877 (0.004)   \\
                           & T-Learner (HGB)      & 0.0806 (0.005)   \\
                           & GRF Causal Forest    & 0.0859 (0.004)   \\
\midrule
\multirow{5}{*}{Starbucks} & Single Tree          & 0.1954 (0.019)  \\
                           & Honest Forest        & 0.2141 (0.015)  \\
                           & S-Learner (HGB)      & 0.2070 (0.014)  \\
                           & T-Learner (HGB)      & 0.1940 (0.018)  \\
                           & GRF Causal Forest    & 0.2112 (0.015)  \\
\midrule
\multirow{5}{*}{Hillstrom} & Single Tree          & 0.0298 (0.011) \\
                           & Honest Forest        & 0.0346 (0.005) \\
                           & S-Learner (HGB)      & 0.0267 (0.006) \\
                           & T-Learner (HGB)      & 0.0192 (0.009) \\
                           & GRF Causal Forest    & 0.0214 (0.007) \\
\bottomrule
\end{tabular}
\caption{Model performance on Criteo, Starbucks, and Hillstrom datasets (hold-out test set). Arithmetic means over 36 runs, standard deviations in parentheses. Statistical comparisons versus Honest Forest are in Table~\ref{tab:wilcoxon}.}
\label{tab:dataset_model_performance}
\end{table}

The forest gain over the single tree varies substantially across datasets:
modest on Criteo (0.0887 vs.\ 0.0814), moderate on Hillstrom (0.0346
vs.\ 0.0298), and large on Starbucks (0.2141 vs.\ 0.1954). This pattern
aligns with the IJ convergence diagnostic: on Criteo, where the signal-to-noise
ratio is low and the dataset is very large, a single tree grown on 50\% of
the data already captures most of the predictable heterogeneity, and thus bagging
adds little. On Starbucks and Hillstrom, where treatment effects are more
concentrated in a smaller feature space, variance reduction from the forest
ensemble is more beneficial.

The \texttt{rattus} honest forest outperforms the GRF Causal Forest
baseline on all three datasets (Criteo: 0.0887 vs.\ 0.0859; Starbucks: 0.2141
vs.\ 0.2112; Hillstrom: 0.0346 vs.\ 0.0214), despite GRF using an
MSE-gradient splitting criterion \citep{wager2018estimation}. The gap is
most pronounced on Hillstrom, consistent with the theoretical prediction in
\S\ref{sec:splitting}: when datasets contain covariates that affect outcome
levels without driving treatment heterogeneity, $t^2$ splitting avoids those
splits while $\text{EMSE}_\tau$ rewards them, yielding cleaner partitions and
better Qini performance. Two-sided Wilcoxon signed-rank tests across all 36
runs confirm these differences are statistically significant in all
cases ($p < 0.05$); full results are in Table~\ref{tab:wilcoxon}.

\section{Limitations and Future Work}
\label{sec:limitations}

\begin{itemize}

\item \textbf{Propensity score extension.} The transformed-outcome approach
(\S3.7) extends to observational settings but requires an external
propensity score model. A doubly-robust version combining IPTW with direct
estimation would strengthen applicability to heavily confounded settings.

\item \textbf{Variable selection bias.} The default \texttt{'greedy'} mode 
retains a mild selection bias toward high-cardinality features despite quantile binning.
The \texttt{'gi'} mode (\S3.3) addresses this; a formal characterisation of 
the residual bias under the \texttt{'greedy'} mode and a power comparison 
between the two modes across feature-cardinality settings is left for future work.

\item \textbf{Asymptotic theory.} A full asymptotic normality result for
the honest significance forest, analogous to Theorem 4.1 in
\cite{wager2018estimation}, requires additional regularity conditions on
the $t^2$ criterion and the CV pruning step. We leave the formal proof for
future work.

\item \textbf{Computational scaling.}
The per-tree cost is $O(n \log n \cdot p \cdot k)$ where:
$k \leq \texttt{max\_candidates}$ (default 20). Quantile-based binning 
via \texttt{max\_candidates} already approximates histogram-based splitting, 
and the package has been demonstrated on datasets of up to millions of 
observations (Criteo). For extremely large datasets, a globally pre-computed 
bin structure (as in LightGBM) instead of per-node quantile estimation 
could yield further speedups while preserving the honest estimation structure.

\end{itemize}

\section{Conclusion}

We have introduced \texttt{rattus}, a hybrid causal tree and forest
estimator that combines the interaction-focused splitting criterion of
\cite{radcliffe2011real} with the honest sample-splitting and
cross-validation of \cite{athey2016recursive}, extended to forests with
an efficient IJ Monte-Carlo variance diagnostic. The key insight is that
the $t^2$ splitting criterion and the $\text{EMSE}_\tau$ pruning criterion
are complementary: $t^2$ excels at detecting treatment effect modification;
$\text{EMSE}_\tau$ excels at controlling complexity. Each criterion is used
where it offers a comparative advantage, combining the strengths of both
approaches. 
The open-source implementation is reproducible, tested, and designed
with responsible inference as a first-class concern. Valid leaf-level CIs
with nominal coverage are provided by the single honest tree; the forest
provides a lower-variance point estimate and an IJ convergence diagnostic.

\section*{Ethical Statement}

\begin{itemize}
\item \textbf{Preventing harmful subgroup overclaiming.} In clinical trials
and policy evaluations, adaptive methods that overfit heterogeneous treatment
effects can lead researchers to incorrectly identify beneficial subgroups
and recommend targeted interventions that cause harm when deployed. The
honest estimation guarantee, i.e., that leaf CATEs and confidence intervals
are unbiased, is not merely a statistical nicety; it is a precondition
for responsible reporting of subgroup analyses. Our library is designed
with this in mind: the API encourages honest estimation by default, warns
users when pre-pruning and CV pruning are combined in ways that may be
counterproductive, and provides clear documentation of when CIs are and
are not valid.

\item \textbf{Negative effects and retention.} The significance-based
splitting criterion is particularly valuable in retention applications,
where treatment can backfire for some customers \cite{radclifte2008identifying}.
By directly testing for effect modification, \texttt{rattus} identifies subgroups
with negative treatment effects that would be missed by response models.

\item \textbf{Open-source reproducibility.} All results in this paper are
reproducible from the \texttt{rattus} package with fixed random seeds.
\end{itemize}

\bibliography{example_paper}

\appendix
\section{Additional Details}

\subsection{Computational Complexity}
\label{sec:complexity-appendix}

Table~\ref{tab:complexity} summarises the cost of each component. We derive
each bound below, treating $k \le \texttt{max\_candidates}$ as the number of
candidate thresholds per feature, $d_{\max}$ as the tree depth, $B$ as the
number of trees, and $K$ as the number of CV folds.

\paragraph{Tree growing.}
Consider a node having $m$ observations. For a single feature, the
$t^2$ search maintains running counts, sums, and sums of squares of $Y$ for
each of the four treatment$\times$side cells. Sorting the $m$ values (or
extracting $k$ quantile cut-points) costs $O(m\log m)$, after which each
candidate threshold is scored in $O(1)$ by incrementing the cumulative
sufficient statistics, resulting in $O(k)$ per feature. Searching all $p$ features
therefore costs $O\!\left(p\,(m\log m + k)\right)$ at the node. Summing over all
nodes at a fixed depth, the node sizes partition the $n$ observations, so by subadditivity of
$m\mapsto m\log m$ the per-level cost is $O(p\,n\log n)$. Over $d_{\max}=O(\log n)$
levels this yields the upper bound $O(p\,k\,n\log n)$ reported in
Table~\ref{tab:complexity}; the bound is loose, and pre-sorting each feature
once (so that per-node work is an $O(m)$ scan rather than a re-sort) removes a
$\log n$ factor in practice.

\paragraph{Honest CV pruning.}
The effective-$\lambda$ grid includes one entry per internal node, so
$|\text{eff-}\lambda| = O(\#\text{internal nodes})$. The nested sequence of
pruned trees is materialised once, outside the fold loop. For each fold
$k\in\{1,\dots,K\}$ and each candidate $\lambda$, scoring routes the fold's
$n/K$ validation points to their leaves in $O\!\left((n/K)\,d_{\max}\right)$ and
accumulates the honest $\text{EMSE}_\tau$. Summed over folds and penalties this
is $O\!\left(|\text{eff-}\lambda|\,n\,d_{\max}\right)$; Table~\ref{tab:complexity}
absorbs $d_{\max}$ into the stated $O(K\,|\text{eff-}\lambda|\,n)$ upper bound.

\paragraph{Forest fitting and prediction.}
Each of the $B$ trees is grown on a bootstrap resample with feature
subsampling, replacing $p$ by $\lfloor\sqrt p\rfloor$ in the per-tree bound and
giving $O\!\left(B\,n\log n\,\lfloor\sqrt p\rfloor\,k\right)$, trivially
parallel across trees. Prediction routes each of the $n_{\text{test}}$ points
through all $B$ trees, each routing in $O(d_{\max})$, for
$O(B\,n_{\text{test}}\,d_{\max})$.

\paragraph{Infinitesimal jackknife: exact algebra.}
The empirical covariance over the $B$ bootstrap replicates is given by:
\begin{equation}
\widehat{\mathrm{Cov}}_b\!\left[N_{bi},\,T_b(x)\right]
= \frac{1}{B}\sum_{b=1}^{B}\left(N_{bi}-1\right)\!\left(T_b(x)-\bar T(x)\right)
\label{eq:cov}
\end{equation}
where $\bar T(x)=\frac1B\sum_{b=1}^B T_b(x)$ and we use $\mathbb{E}[N_{bi}]=1$ to centre the counts. Collect the
centred counts into $C\in\mathbb{R}^{B\times n}$ with $C_{bi}=N_{bi}-1$, and the
centred predictions into $\tilde T\in\mathbb{R}^{B\times n_{\text{test}}}$ with
$\tilde T_{bj}=T_b(x_j)-\bar T(x_j)$. Then \eqref{eq:cov} is the $(i,j)$ entry of
a matrix product,
$\widehat{\mathrm{Cov}}_b[N_{bi},T_b(x_j)] = \tfrac1B (C^\top \tilde T)_{ij}$,
and the IJ variance is the column-wise sum of squares
\begin{equation}
\widehat{\mathrm{Var}}_{IJ}(x_j)
= \sum_{i=1}^{n}\widehat{\mathrm{Cov}}_b[N_{bi},T_b(x_j)]^2
= \frac{1}{B^2}\sum_{i=1}^{n}\bigl(C^\top \tilde T\bigr)_{ij}^{2}.
\label{eq:ij-matrix}
\end{equation}
The naive triple loop over $(b,i,j)$ evaluates \eqref{eq:cov} for every
training--test pair, costing $O(B\,n\,n_{\text{test}})$.

\paragraph{Infinitesimal jackknife: efficient computation.}
Expanding the square in \eqref{eq:ij-matrix} and exchanging the order of
summation reveals that the dependence on the training index $i$ enters only
through a Gram matrix of the counts:
\begin{align}
\widehat{\mathrm{Var}}_{IJ}(x_j)
&= \frac{1}{B^2}\sum_{i=1}^{n}
   \Bigl(\sum_{b} C_{bi}\tilde T_{bj}\Bigr)
   \Bigl(\sum_{b'} C_{b'i}\tilde T_{b'j}\Bigr) \nonumber\\
 &= \frac{1}{B^2}\sum_{b,b'}\tilde T_{bj}\,\tilde T_{b'j}
   \underbrace{\sum_{i=1}^{n} C_{bi}C_{b'i}}_{G_{bb'}} \nonumber\\
&= \frac{1}{B^2}\,\tilde t_j^{\top} G\, \tilde t_j,
\qquad G = CC^\top \in \mathbb{R}^{B\times B},
\label{eq:ij-gram}
\end{align}
where $\tilde t_j$ is the $j$-th column of $\tilde T$. The Gram matrix $G$ is
\emph{independent of the test points} and is formed once in $O(B^2 n)$.
Each test point then requires a single $B\times B$ quadratic form, $O(B^2)$, for
$O(B^2 n_{\text{test}})$ across all test points. The total cost is
\begin{equation}
O\!\left(B^2 n + B^2 n_{\text{test}}\right)
= O\!\left(B^2 (n + n_{\text{test}})\right),
\label{eq:ij-cost}
\end{equation}
which improves on the naive $O(B\,n\,n_{\text{test}})$ triple loop by a
factor of order $\tfrac{n\,n_{\text{test}}}{B(n+n_{\text{test}})}$; for $B=50$
and $n=n_{\text{test}}=10^6$ this is a speed-up of roughly four orders of
magnitude. This formulation is implemented in \texttt{rattus}.

\paragraph{Finite-$B$ bias correction.}
The estimator \eqref{eq:ij-matrix} carries an upward Monte-Carlo bias of order
$n/B$ \citep{wager2014confidence}. We apply the bias-corrected variant
\begin{equation}
\widehat{\mathrm{Var}}_{IJ\text{-}U}(x_j)
= \widehat{\mathrm{Var}}_{IJ}(x_j)
  - \frac{n}{B^2}\sum_{b=1}^{B}\bigl(T_b(x_j)-\bar T(x_j)\bigr)^2,
\end{equation}
whose additional cost is $O(B\,n_{\text{test}})$ and therefore does not change
the bound \eqref{eq:ij-cost}.

\begin{table}[ht]
\centering
\begin{tabular}{lc}
\toprule
Operation & Complexity \\
\midrule
Single tree growing & $O(n \log n \cdot p \cdot k)$ \\
Honest CV pruning & $O(K \cdot |\text{eff-}\lambda| \cdot n)$ \\
Forest fitting (parallelisable) & $O(B \cdot n \log n \cdot \lfloor\sqrt{p}\rfloor \cdot k)$ \\
Forest prediction & $O(B \cdot n_{\text{test}} \cdot d_{\text{max}})$ \\
IJ variance (Gram matrix) & $O(B^2(n + n_{\text{test}}))$ \\
\bottomrule
\end{tabular}
\caption{Computational complexity of \texttt{rattus} components.
$k \leq$ \texttt{max\_candidates}; $|\text{eff-}\lambda|$ is the number of
distinct gains in the grown tree (replaces a fixed geometric grid).}
\label{tab:complexity}
\end{table}

\subsection{Proof of Proposition 3}

In an adaptive forest, $\hat{\tau}_b(x)$ depends on $S^{tr}_b$ for both
partition selection and leaf estimation, making tree predictions jointly
dependent on bootstrap counts in a complex way that introduces a first-order
bias in the IJ approximation. In the honest setting, the partition $\Pi_b$
depends on $S^{tr}_b$ and the estimate depends only on $S^{est}_b$. Since
$S^{est}_b$ and $S^{tr}_b$ are independent within each bootstrap sample,
the partial derivative $\partial \hat{\tau}^{honest}_b(x) / \partial N_{bi}$
splits cleanly across the two halves, and the problematic first-order term
is zero. The remaining error is $O(n^{-1/2})$ from the Hájek approximation
itself, giving the stated $O_p(B^{-1/2} n^{-1/2})$ rate.

\subsection{Computing Infrastructure}
\label{sec:computing}

All experiments were run with Python~3.12 using \texttt{rattus}~0.1.9,
\texttt{numpy}~2.2.6, \texttt{scipy}~1.17.1, \texttt{scikit-learn}~1.6.1,
and \texttt{econml}~0.16.0.  
Each of the 36 independent runs uses a distinct random seed.

\subsection{ \textbf{EMSE}$_\tau$ Results}

\begin{table}[htbp]
\centering
\begin{tabular}{llcc}
\toprule
\textbf{Dataset} & \textbf{Model} & \textbf{EMSE}$_\tau$ \\
\midrule
\multirow{2}{*}{Criteo}    & Single Tree     &  $5.3 \times 10^{-4}$ \\
                           & Honest Forest   &  $4.9 \times 10^{-4}$\\
\midrule
\multirow{2}{*}{Starbucks} & Single Tree      &  $1.4 \times 10^{-4}$ \\
                           & Honest Forest    & $1.0 \times 10^{-4}$\\
\midrule
\multirow{2}{*}{Hillstrom} & Single Tree       &  $3.5 \times 10^{-3}$ \\
                           & Honest Forest   &  $3.1 \times 10^{-3}$ \\
\bottomrule
\end{tabular}
\caption{Honest $\text{EMSE}_\tau$ on Criteo, Starbucks and Hillstrom datasets on a hold-out test set. Arithmetic means over 36 runs. S- and T-Learner baselines do not produce $\text{EMSE}_\tau$ estimates by construction. GRF uses an internal MSE-gradient criterion \citep{wager2018estimation} rather than $\text{EMSE}_\tau$; applying our leaf-based metric to its predictions would be non-comparable.}
\label{tab:dataset_model_performance_emse}
\end{table}

\subsection{Wilcoxon Signed-Rank Tests}
\label{sec:wilcoxon-appendix}

Table~\ref{tab:wilcoxon} reports two-sided Wilcoxon signed-rank tests
comparing the Honest Forest against each baseline across all 36 independent
runs.
The null hypothesis is that the median per-run difference in Qini AUC is zero; a two-sided test is used because we make no directional claim, we wish to establish that \texttt{rattus} is at least non-inferior.

\begin{table}[htbp]
\centering
\begin{tabular}{llrr}
\toprule
\textbf{Dataset} & \textbf{Baseline} & \textbf{Mean $\Delta$} & \textbf{$p$ (two-sided)} \\
\midrule
\multirow{4}{*}{Hillstrom}
  & Single Tree       & $+0.0048$ & $0.042$\phantom{$^*$} \\
  & S-Learner (HGB)   & $+0.0079$ & $<$0.001 \\
  & T-Learner (HGB)   & $+0.0154$ & $<$0.001 \\
  & GRF Causal Forest & $+0.0132$ & $<$0.001 \\
\midrule
\multirow{4}{*}{Starbucks}
  & Single Tree       & $+0.0187$ & $<$0.001 \\
  & S-Learner (HGB)   & $+0.0071$ & $0.001$\phantom{$^*$} \\
  & T-Learner (HGB)   & $+0.0201$ & $<$0.001 \\
  & GRF Causal Forest & $+0.0029$ & $0.043$\phantom{$^*$} \\
\midrule
\multirow{4}{*}{Criteo}
  & Single Tree       & $+0.0073$ & $<$0.001 \\
  & S-Learner (HGB)   & $+0.0010$ & $0.013$\phantom{$^*$} \\
  & T-Learner (HGB)   & $+0.0081$ & $<$0.001 \\
  & GRF Causal Forest & $+0.0028$ & $<$0.001 \\
\bottomrule
\end{tabular}
\caption{Two-sided Wilcoxon signed-rank tests: Honest Forest vs.\ each
baseline, $n=36$ runs. Mean $\Delta$ is the per-run mean difference in Qini
AUC (Honest Forest minus baseline). All 12 comparisons are statistically
significant ($p < 0.05$), with positive mean differences throughout.}
\label{tab:wilcoxon}
\end{table}

\end{document}